\documentclass[11pt]{article}
\usepackage[a4paper, margin=1in]{geometry}
\usepackage[utf8]{inputenc}
\usepackage{authblk}

\usepackage{graphicx,caption}
\usepackage{amsmath,amssymb,amsthm}
\usepackage{bm}
\usepackage{braket}
\usepackage{color}
\usepackage{enumerate}
\usepackage{empheq}
\usepackage{booktabs}
\usepackage{url}
\usepackage[
    colorlinks=true, 
    urlcolor=blue, 
    citecolor=blue, 
    anchorcolor=red
]{hyperref}
\usepackage{tikz}
\usetikzlibrary{quantikz2}
\numberwithin{equation}{section}
\allowdisplaybreaks
\usepackage{mathtools}
\usepackage{xspace}
\newtheorem{theorem}{Theorem}
\newtheorem{lemma}{Lemma}

\newtheorem{definition}{Definition}
\newtheorem{problem}{Problem}
\newcommand{\e}{\mathrm{e}}

\newcommand{\Tr}{\mathrm{Tr}}

\newcommand{\C}{\mathbb{C}}

\newcommand*{\bigo}{\ensuremath{O}\xspace}

\newcommand{\ketbra}[2]{|#1\rangle\! \langle #2|}
\DeclarePairedDelimiterXPP\lnorm[1]{}{\lVert}{\rVert}{}{#1}
\DeclarePairedDelimiter{\opnorm}{\lVert}{\rVert}

\DeclarePairedDelimiter{\ceil}{\lceil}{\rceil}
\NewDocumentCommand{\PARITY}{o}{%
\ensuremath{%
    \textsc{PARITY}\IfNoValueTF{#1}{}{_{\,\mathclap{#1}\,}}%
}\xspace%
}
\newcommand*{\PAR}{\ensuremath{\textsc{PAR}}}
\NewDocumentCommand{\OR}{o}{%
\ensuremath{%
    \textsc{OR}\IfNoValueTF{#1}{}{_{#1}}%
}\xspace%
}
\NewDocumentCommand{\PARITYOR}{oo}{%
\ensuremath{
    \PARITY[#1]\,\circ\,\OR[#2]%
}\xspace%
}

\title{Sparsity-dependent Complexity Lower Bound of\\ Quantum Linear System Solvers}

\author[1,2]{Hitomi Mori\thanks{\href{mailto:hitomi.mori@quantinuum.com}{hitomi.mori@quantinuum.com}}}

\author[1,3]{Yuta Kikuchi}

\author[4]{Marcello Benedetti}

\author[4]{Matthias Rosenkranz}

\affil[1]{Quantinuum K.K., Tokyo, Japan}
\affil[2]{Centre for Quantum Technologies, National University of Singapore, Singapore}
\affil[3]{Center for Interdisciplinary Theoretical and Mathematical Sciences, RIKEN, Wako, Japan}
\affil[4]{Quantinuum Ltd., London, United Kingdom}

\begin{document}

\date{January 26, 2026}

\maketitle

\vspace{-.5cm}

\begin{abstract}
Quantum linear system (QLS) solvers are a fundamental class of quantum algorithms used in many potential quantum computing applications, including machine learning and solving differential equations.
The performance of quantum algorithms is often measured by their query complexity, which quantifies the number of oracle calls required to access the input. The main parameters determining the complexity of QLS solvers are the condition number $\kappa$ and sparsity $s$ of the linear system, and the target error $\epsilon$.
To date, the best known query-complexity lower bound is $\Omega(\kappa\log(1/\epsilon))$, which establishes the optimality of the most recent QLS solvers.
The original proof of this lower bound is attributed to Harrow and Kothari, but their result is unpublished.
Furthermore, when discussing a more general lower bound including the sparsity $s$ of the linear system, it has become folklore that it should read as $\Omega( \kappa \sqrt{s}\log(1/\epsilon))$. 
In this work, we establish the rigorous lower bound capturing the sparsity dependence of QLS.
We prove the lower bound of $\Omega(\kappa\sqrt{s})$ for any quantum algorithm that solves QLS with constant error via a reduction of the bounded-error \PARITYOR problem.
While the dependence on all parameters $\kappa,s,\epsilon$ remains an open problem, our result provides a crucial stepping stone toward the complete characterization of QLS complexity.
\end{abstract}

\section{Introduction}
The study of query complexity lower bounds is key to the understanding of the theoretical limits of computation. When the complexity lower bound of a certain problem matches the complexity upper bound of a known algorithm, this establishes the asymptotic optimality of that algorithm. 
The quantum linear system (QLS) problem, which is widely studied in the quantum computing community, consists of finding an approximation $\ket{\tilde{x}}$ of the normalized quantum state $\ket{x}=A^{-1}\ket{b}/\lnorm{A^{-1}\ket{b}}$ that satisfies $\lnorm{\ket{\tilde{x}} - \ket{x}} \le \epsilon$ for some  $\epsilon>0$.
Here, $A$ is an invertible $N \times N$ matrix, $\ket{b}$ is a quantum state, and $\lnorm{\cdot}$ is the $\ell^2$-norm.
Thus $\ket{\tilde{x}}$ encodes an approximate solution to the linear system of equations $Ax = b$. 
The most relevant parameters that determine the complexity of QLS are the condition number $\kappa$ and sparsity $s$ of the matrix $A$, and the target error $\epsilon$.

QLS was introduced and solved almost two decades ago by Harrow, Hassidim and Lloyd~\cite{Harrow_2009}, and since then new quantum algorithms have been discovered and refined to the point that they achieve query complexity $\bigo(\kappa \log(1/\epsilon))$ with respect to the block-encoding oracle~\cite{Costa_2021,Dalzell_2024,Cunningham_2024,Low_2026}. This upper bound is considered optimal in both the condition number $\kappa$ and the error $\epsilon$ because it matches a `known' lower bound of $\Omega(\kappa \log (1/\epsilon))$. The original proof of this lower bound is attributed to Harrow and Kothari, but the result is unpublished. According to several authors, a manuscript has been in preparation since at least 2018, as evidenced by the bibliographic entries in their research papers, e.g., Refs.~\cite{Low_2019,Costa_2021,Dalzell_2022,Costa_2023,Wang_2024,An2025,Chen_2025}. Furthermore, when discussing a more general lower bound that includes the sparsity $s$ of the linear system, it has become folklore that it should read as $\Omega( \kappa \sqrt{s} \log(1/\epsilon))$ also due to Harrow and Kothari. For example, this is mentioned in a recent review on QLS solvers~\cite{Morales_2024}, although the authors acknowledge that the role of sparsity remains an open question.

The aim of this work is to clarify some of the claims in the aforementioned literature. First, we give a proof of the $\Omega(\kappa \log(1/\epsilon))$ lower bound via a reduction from the \PARITY problem in the \emph{unbounded} error setting. This is not a new result, because a proof strategy was discussed in~\cite[Appendix A]{Costa_2023}. However, we believe a detailed and explicit derivation will be useful to researchers in quantum algorithms. Second, we prove the $\Omega(\kappa\sqrt{s})$ lower bound via a reduction from the \PARITYOR problem in the \emph{bounded} error setting. With this setting, we can only establish a lower bound for $\epsilon=\bigo(1)$, so the $\epsilon$-dependence does not appear in the expression. To the best of our knowledge, this result is new, and there exists a nearly matching upper bound of $O((\kappa\sqrt{s})^{1+o(1)})$ by Ref. \cite{Low_2019}. We believe it is an important step towards understanding the role of sparsity in the query complexity of the QLS problem and towards a joint lower bound in all three parameters $\kappa$, $s$, and $\epsilon$.

The remainder of the paper is organized as follows. In Section~\ref{sec:preliminaries}, we provide background material relevant to our proofs, such as the data access model, the \PARITY problem and the \PARITYOR problem. 
In Section~\ref{sec:precision}, we give the detailed proof for the $\Omega(\kappa \log (1/\epsilon))$ lower bound.
Our main result is presented in Section~\ref{sec:sparsity}, where we prove the sparsity-dependent $\Omega(\kappa\sqrt{s})$ lower bound.
Lastly, we discuss the challenges and possible extensions towards a joint lower bound that captures all parameters $\kappa,s,\epsilon$ in Section~\ref{sec:discussion}.

\section{Preliminaries}
\label{sec:preliminaries}

For vectors $x$, we use the $\ell^2$-norm $\lnorm{x}=\sqrt{x\cdot x}$ and denote by $\ket{x}$ a quantum state encoding $x$ satisfying $\lnorm{\ket{x}}=\sqrt{\langle x|x\rangle}=1$. For matrices $A$, we use the spectral norm $\opnorm{A}=\max_{\lnorm{x}\ne0}\lnorm{Ax}/\lnorm{x}$.
For integer variables $a$ and $b$, we define $[a,b]:=\{a,a+1,\dots, b\}$, and we simply write $[a]:=[1,a]$. We use the Bachmann-Landau notation $\Omega$ and $O$ for asymptotic lower and upper bounds, respectively, and $\Theta$ for a simultaneous lower and upper bound (tight bound). Oracles are denoted by calligraphic letters such as $\mathcal{P}$ and $\mathcal{O}$.

\begin{problem}[Quantum Linear System]
\label{prob:QLS}
    Let $A\in \mathbb{C}^{N\times N}$ be an invertible matrix with condition number $\kappa:=\opnorm{A}\,\opnorm{A^{-1}}$ and sparsity $s$, i.e., it has at most $s$ non-zero entries in each row and column. Let $\ket{b}$ be a quantum state and let $\epsilon>0$ be the target error.
    Prepare an approximation $\ket{\tilde{x}}$ of the normalized quantum state 
    \begin{align}\label{eq:QLS_x}
        \ket{x}=\frac{A^{-1}\ket{b}}{\lnorm{A^{-1}\ket{b}}}
    \end{align}
    that satisfies $\lnorm{\ket{\tilde{x}} - \ket{x}} \le \epsilon$.
\end{problem}
\noindent Note that by the singular-value decomposition we have $A=\sum_{i=1}^{N} \sigma_i \ketbra{v_i}{w_i}$ and the notation in Eq.~\eqref{eq:QLS_x} means $A^{-1}\ket{b}=\sum_{i=1}^{N}\sigma_i^{-1} \braket{v_i|b} \ket{w_i}$.

To discuss the query complexity of the QLS problem we begin by specifying the access model for an input matrix~$A$.
In this paper, we assume that the cost of preparing the initial state $\ket{b}$ is negligible compared to the cost of computing $A^{-1}$, and we therefore focus on query access to the matrix $A$.

\begin{definition}[Sparse-matrix oracles \cite{Berry_2015,Berry_2012}]
\label{def:A_oracle}
    Given a matrix $A$ which is $s$-sparse, a quantum sparse-matrix-access $\mathcal{P}_A$ is given by a pair of oracular functions
    \begin{align}
    \mathcal{P}_A:=\left(\mathcal{P}_A^{\rm pos}, \mathcal{P}_A^{\rm val}\right)
    \end{align}
    where $\mathcal{P}_A^{\rm pos}$ and $\mathcal{P}_A^{\rm val}$ specify the positions of the (potentially) non-zero entries of $A$ and the values of those entries, respectively, i.e., for $i,j \in[N], \nu \in[s]$, and $z \in\{0,1\}^*$,
    \begin{align}
    \label{eq:position_oracle}
    &\mathcal{P}_A^{\rm pos}:|i, \nu\rangle \mapsto|i, j(i, \nu)\rangle,\\
    &\mathcal{P}_A^{\rm val}:|i, j, z\rangle \mapsto\left|i, j, A_{i, j} \oplus z\right\rangle.
    \end{align}
    Here $A_{i, j} \oplus z \in\{0,1\}^*$ denotes a bit string of arbitrary fixed length that encodes the value $A_{i, j} \in \mathbb{C}$. The matrix element $A_{i,j}$ is represented in binary and $\oplus$ denotes bitwise XOR. The function $j(i, \nu)$ returns the column index of the $\nu$-th non-zero entry in row $i$ or the column index of any zero entry in row $i$ if there are fewer than $\nu$ non-zero entries in row $i$.
\end{definition}

In our proofs, we leverage known tight complexity bounds of Boolean problems. 
The idea of using a reduction from a Boolean problem to obtain a lower bound for QLS has been used in Refs. \cite{Costa_2023,Orsucci2021}.
A Boolean problem takes $n$ Boolean variables $X=(x_1,\dots,x_n)\in\{0, 1\}^n$ and asks to return $f(X)\in\{0,1\}$ for a Boolean function $f$. Boolean functions are called \textit{partial} when they are only defined on a subset of $\{0, 1\}^n$, and otherwise are called \textit{total}. A quantum algorithm solves these problems by preparing a quantum state $\rho$ and then performing a projective measurement with elements $\{\Pi, I-\Pi\}$, where an outcome of 1 (corresponding to the projection with $\Pi$) suggests that $f(X)=1$, while an outcome of 0 (corresponding to the projection with $I -\Pi$) suggests that $f(X)=0$.

\begin{problem}[Bounded/unbounded-error quantum problem]
\label{def:boolean_problem}
    Let $f$ be a total or partial Boolean function and ${\Pi}$ be a projection operator.
    For any allowed input $X$ and for some bias $\gamma > 0$, create a density matrix $\rho$ such that
    \begin{align}
    \label{eq:unbounded_parity_even}
        \Tr[\rho{\Pi}] \geq \frac{1}{2} + \gamma
    \end{align}
    if $f(X)=1$, and
    \begin{align}
    \label{eq:unbounded_parity_odd}
        \Tr[\rho{\Pi}] \leq \frac{1}{2} - \gamma
    \end{align}
    if $f(X)=0$.
    We say it is a bounded-error quantum problem when $\gamma =\Omega(1)$. We say it is an unbounded-error quantum problem when $\gamma$ can be any positive value.
\end{problem}

The first Boolean problem of interest is called \PARITY.
Provided access to $X=(x_1,\dots,x_n)$, the \PARITY[n] problem asks to return the value $f_{\PAR}(X):=\oplus_{t=1}^{n}x_t$.
In our quantum computing setup, we consider the access model
\begin{align}
\label{eq:access_parity}
    \mathcal{O}^{(1)}_{X}\ket{k,t} = \ket{k\oplus x_t,t}
    \quad\text{for }
    k\in\{0,1\} \text{ and } t\in[n].
\end{align}
The following two lemmata show that the quantum query complexity for the \PARITY problem with access to the oracle $\mathcal{O}_X^{(1)}$ is the same in the bounded- and unbounded-error setting.
\begin{lemma}[Tight bound of bounded-error \PARITY \cite{Farhi_1998}]\label{lem:parity}
    The bounded-error quantum query complexity $\mathsf{Q}$ of \PARITY[n] is $\mathsf{Q}(\PARITY[n]) = \Theta(n)$.
\end{lemma}

\begin{lemma}[Tight bound of unbounded-error \PARITY \cite{Farhi_1998}]\label{lem:parity_unbounded}
    The unbounded-error quantum query complexity $\mathsf{UQ}$ of \PARITY[n] is $\mathsf{UQ}(\PARITY[n]) = \Theta(n)$.
\end{lemma}

The second Boolean function of interest is the composition of the \PARITY function and a partial \OR function, which we denote as \PARITYOR. Provided access to $X=(X_1,\dots,X_n)=(x_{1,1},x_{1,2},\dots,x_{1,m},x_{2,1},\dots,x_{n,m})$, the \PARITYOR[n][m] problem asks to return the value $f_{\PAR\text{-}\OR}(X):=f_{\PAR}(f_{\OR}(X_1),\dots,f_{\OR}(X_n))$.
The partial \OR[m] function assumes that at most one bit of each input bit string $X_t=(x_{t,1},\dots,x_{t,m})$ is non-zero and returns the value $f_{\mathrm{OR}}(X_t)=x_{t,1}\lor\cdots\lor x_{t,m}$.
This is formally termed ``promise OR'' problem, but we omit ``promise'' throughout the paper.
In our quantum computing setup, we consider the following access model
\begin{align}
\label{eq:access_party_or}
    \mathcal{O}^{(2)}_{X}\ket{k,i,t} = \ket{k\oplus x_{t,i},i,t}
    \quad\text{ for }
    k\in\{0,1\}, i\in[m] \text{ and } t\in[n].
\end{align}

Using Ref.~\cite[Theorem 1.5]{reichardtReflectionsQuantumQuery2010} to compose the bounded-error bounds of \PARITY, $\Theta(n)$~\cite{Farhi_1998}, and \OR, $\Theta(\sqrt{m})$~\cite{Grover_1996,Beals2001}, we obtain the following result.
Note that Ref.~\cite[Theorem 1.5]{reichardtReflectionsQuantumQuery2010} is also applicable to partial functions. See Refs.~\cite[Theorem 1.3]{reichardtReflectionsQuantumQuery2010} and \cite[Theorem 11]{Hoyer_2007}.

\begin{lemma}[Tight bound of bounded-error \PARITYOR \cite{Farhi_1998,Grover_1996,reichardtReflectionsQuantumQuery2010}]\label{lem:parity_or}
    The bounded-error quantum query complexity $\mathsf{Q}$ of \PARITYOR[n][m] is $\mathsf{Q}(\PARITYOR[n][m]) = \Theta(n\sqrt{m})$.
\end{lemma}

It is known that the unbounded-error \OR[m] problem is solved with a single query~\cite[Section 5.3.2]{Montanaro_2011}, i.e. $\mathsf{UQ}(\OR[m])=1$. We show that this leads to a $\Theta(n)$ query complexity for the unbounded-error \PARITYOR[n][m] problem.
\begin{lemma}[Tight bound of unbounded-error \PARITYOR]
\label{lem:parity_or_unbounded}
    The unbounded-error quantum query complexity $\mathsf{UQ}$ of \PARITYOR[n][m] is $\mathsf{UQ}(\PARITYOR[n][m]) = \Theta(n)$.
\end{lemma}
\begin{proof}
The unbounded-error quantum query complexity of a (partial) Boolean function $f$ is given as $\mathsf{UQ}(f) = \ceil{\text{udeg}(f)/2}$, where $\text{udeg}(f)$ is the lowest degree among all unbounded error polynomials for $f$ \cite[Theorem 3.4]{Montanaro_2011}.
The same degree can also be expressed as the threshold degree $\text{sdeg}(f)$, which is defined as the minimum degree of polynomials that sign-represent $f$, so we have $\text{udeg}(f)=\text{sdeg}(f)$ \cite[discussion after Definition 2.1]{Montanaro_2011}.
The threshold degree of the composed function \PARITYOR satisfies $\text{sdeg}(\PARITY\circ \OR) = \text{sdeg}(\PARITY)\, \text{sdeg}(\OR)$ \cite[Theorem 1.1]{Sherstov_2009}.
Combining these, we have $\mathsf{UQ}(\PARITYOR[n][m]) = \ceil{\text{sdeg}(\PARITYOR[n][m])/2} = \ceil{\text{sdeg}(\PARITY[n])\, \text{sdeg}(\OR[m])/2} = \Theta(\mathsf{UQ}(\PARITY[n])\mathsf{UQ}(\OR[m]))=\Theta(n)$.
\end{proof}

We give an alternative proof in Appendix~\ref{app:composition}, which provides a more intuitive view of the tight complexity bound.

\section{Lower bound in condition number and target error}
\label{sec:precision}

Here we give a detailed derivation of the query-complexity lower bound of QLS in $\kappa$ and $\epsilon$ via a reduction from the unbounded-error \PARITY problem. 

\begin{theorem}[Lower bound of QLS with target error \cite{Costa_2023}]
\label{lem:qls_unbounded}
    Let the {\rm QLS} problem be defined as Problem~\ref{prob:QLS} with the sparsity $s$ fixed to a constant and the target error $0 < \epsilon\le1/11$. The query complexity of {\rm QLS} is $\Omega(\kappa\log (1/\epsilon))$.
\end{theorem}

\begin{proof}
Our goal is to lower-bound the number of queries to the sparse access model in Def.~\ref{def:A_oracle}.
Based on Refs.~\cite{Harrow_2009, Costa_2023}, we define matrices
\begin{align}
    \label{eq:a_def_parity}
    &A := I - \delta^{1/n}B ,
    \\
    \label{eq:b_def_parity}
    &B :=
    \sum_{t=1}^{n} \Big(
        U_t\otimes \ketbra{t+1}{t}_{\rm c}
        + I\otimes \ketbra{t+n+1}{t+n}_{\rm c}
        + U_{n-t+1}\otimes \ketbra{t+2n+1}{t+2n}_{\rm c}
    \Big),\\
    &U_t :=(I \otimes \bra{t}_{\rm c})\mathcal{O}_X^{(1)}(I\otimes\ket{t}_{\rm c}),\label{eq:u_def_parity}
\end{align}
with $0<\delta<1$ so that the matrix $A$ is invertible. The parameter $\delta$ will be related to the target error $\epsilon$ later, such that the corresponding QLS problem solves the unbounded-error \PARITY problem.
With $\mathcal{O}^{(1)}_X$ given by Eq.~\eqref{eq:access_parity} it can be verified that $U_t=\sum_{k\in\{0,1\}}\ketbra{k\oplus x_t}{k}$. 
The clock register indicated by the subscript `c' in the unitary operator $B$ is defined modulo $3n$.
An exact QLS solver with the initial state $\ket{0}\ket{1}_{\rm c}$ creates the state
\begin{align}
\label{eq:inverse_parity}
\begin{split}
    \ket{\Psi}
    &:=\frac{A^{-1}\ket{0}\ket{1}_{\rm c}}{\lnorm*{A^{-1}\ket{0}\ket{1}_{\rm c}}}
    =
    \frac{\sum_{t=0}^{\infty}\delta^{t/n}B^t\ket{0}\ket{1}_{\rm c}}{\lnorm*{A^{-1}\ket{0}\ket{1}_{\rm c}}}
    \\
    &=
    c\sum_{t=0}^{n-1}\delta^{t/n}\Big(
        \ket{\oplus_{j=1}^tx_j}\ket{t+1}_{\rm c}
        \\
        &\qquad
        + \delta \ket{\oplus_{j=1}^nx_j}\ket{t+n+1}_{\rm c}
        \\
        &\qquad
        + \delta^2 \ket{\oplus_{j=1}^{n-t}x_j}\ket{t+2n+1}_{\rm c}
    \Big),
\end{split}
\end{align}
with $c=\left[\left(1-\delta^3\right)\lnorm*{A^{-1}\ket{0}\ket{1}_{\rm c}}\right]^{-1}$.
The QLS solver's target error $\epsilon$ will be discussed later.

Provided the state in Eq.~\eqref{eq:inverse_parity}, we consider the following measurement and post-processing.
First, we measure the clock register in the computational basis. If the measurement outcome indicates $t\in[n+1,2n]$, we accept the state of the first register since it corresponds to the correct parity $\ket{\oplus_{j=1}^nx_j}$. With an exact QLS solver the probability of obtaining this outcome is
\begin{align}
    p
    =
    \frac{\delta^{2}}{1+\delta^{2}+\delta^{4}}.
\end{align}
On the other hand, if we obtain any other measurement outcome on the clock register, we randomly output 0 or 1 with probability $1/2$ each.
This protocol amounts to creating the single-qubit density matrix
\begin{align}
    \rho 
    := 
    \Tr_{\rm c}[(I\otimes{\Pi}_{\rm c})\ketbra{\Psi}{\Psi}]
    +
    \Tr[(I\otimes (I -{\Pi}_{\rm c}))\ketbra{\Psi}{\Psi}]\frac{I}{2},
\end{align}
where we defined the projector ${\Pi}_{\rm c}:=\sum_{t=n+1}^{2n}\ketbra{t}{t}_{\rm c}$ and $\Tr_{\rm c}$ denotes the partial trace over the clock register.
This state $\rho$ corresponds to the one in Problem~\ref{def:boolean_problem}.
Finally, we measure the projector $\Pi:=\ketbra{1}{1}$ on $\rho$ to determine the parity.
When $\oplus_{j=1}^n x_j=1$, we have
\begin{align}
    \Tr[\rho \Pi]
    &=
    p \, \Tr\left[{\Pi}\ketbra{\oplus_{j=1}^n x_j}{\oplus_{j=1}^n x_j}\right]+
    \frac{1 - p}{2} 
    =
    \frac{1}{2} + \frac{p}{2}.
\end{align}
Similarly, when $\oplus_{j=1}^n x_j=0$, we have
\begin{align}
    \Tr[\rho \Pi]
    =
    \frac{1}{2} - \frac{p}{2}.
\end{align}
Thus, there exists $\gamma = \frac{p}{2} = \frac{\delta^{2}}{2(1+\delta^{2}+\delta^{4})}$ such that Eqs.~\eqref{eq:unbounded_parity_even} and~\eqref{eq:unbounded_parity_odd} hold.

To assess the influence of the target error $\epsilon$, we let $\ket{\Psi'}$ be a quantum state such that $\lnorm{\ket{\Psi}-\ket{\Psi'}}\le\epsilon$.
The erroneous density matrix is expressed as
\begin{align}
    \rho'
    := 
    \Tr_{\rm c}[(I\otimes{\Pi}_{\rm c})\ketbra{\Psi'}{\Psi'}]
    +
    \Tr[(I\otimes(I -{\Pi}_{\rm c}))\ketbra{\Psi'}{\Psi'}] \frac{I}{2},
\end{align}
and the probability of measuring outcome 1 (projection by $\Pi$) is
\begin{align}
    \Tr[\rho'{\Pi}]
    &=\lnorm*{(\Pi\otimes\Pi_{\rm c})\ket{\Psi'}}^2
    +
    \frac{1-\lnorm*{(I\otimes \Pi_{\rm c})\ket{\Psi'}}^2}{2}.
\end{align}
When $\oplus_{j=1}^n x_j=1$, using the inequalities 
\begin{align}
    \lnorm*{(\Pi\otimes\Pi_{\rm c})\ket{\Psi'}}
    &\ge{\lnorm*{(\Pi\otimes\Pi_{\rm c})\ket{\Psi}}
    -
    \lnorm*{(\Pi\otimes\Pi_{\rm c})(\ket{\Psi}-\ket{\Psi'})}}
    \ge {\sqrt{p}-\epsilon}
\end{align}
and
\begin{align}
    \lnorm*{(I\otimes\Pi_{\rm c})\ket{\Psi'}}
    &\le{\lnorm*{(I\otimes\Pi_{\rm c})\ket{\Psi}} 
    + \lnorm*{(I\otimes(I- \Pi_{\rm c}))(\ket{\Psi'}-\ket{\Psi})}}
    \le \sqrt{p}+\epsilon,
\end{align}
one can lower-bound the probability as
\begin{align}
    \Tr[\rho'{\Pi}]
    &\ge\frac{1}{2}+\frac{p}{2}-3\sqrt{p}\epsilon+\frac{\epsilon^2}{2}.
\end{align}
Similarly, when $\oplus_{j=1}^n x_j=0$, we have
\begin{align}
    \Tr[\rho'{\Pi}]
    \le \frac{1}{2}-\frac{p}{2}-\sqrt{p}\epsilon+\frac{\epsilon^2}{2}.
\end{align}
Note that $\delta$ needs to be chosen such that $\frac{p}{2}-3\sqrt{p}\epsilon+\frac{\epsilon^2}{2}$ and $\frac{p}{2}+\sqrt{p}\epsilon-\frac{\epsilon^2}{2}$ are both positive under the condition $0<\epsilon\le1/11$ so that Eqs.~\eqref{eq:unbounded_parity_even} and~\eqref{eq:unbounded_parity_odd} for the unbounded-error Problem~\ref{def:boolean_problem} can be satisfied. 
For instance, if we take $\delta=8\epsilon$, it can be verified that there exists $\gamma>0$ such that Eqs.~\eqref{eq:unbounded_parity_even} and~\eqref{eq:unbounded_parity_odd} hold.

We have reduced the unbounded-error \PARITY[n] to the QLS problem and, from Lemma~\ref{lem:parity_unbounded}, the former has query-complexity lower bound of $\Omega(n)$ in terms of access to $\mathcal{O}^{(1)}_X$.
From Eqs.~\eqref{eq:a_def_parity}, \eqref{eq:b_def_parity} and \eqref{eq:u_def_parity} one access to the sparse matrix oracle $\mathcal{P}_A$ requires at most two queries to controlled $\mathcal{O}^{(1)}_X$, analogously to the construction given in Appendix \ref{app:sparse_oracle}, so the query complexity of the QLS solver is also lower bounded by $\Omega(n)$.
The condition number $\kappa$ of $A$ satisfies $\kappa\le 2/(1-\delta^{1/n})=\Theta(n/\log(1/\epsilon))$ because the unitary $B$ has eigenvalues on the unit circle.
Therefore, we have the lower bound of the query complexity $\Omega(\kappa\log (1/\epsilon))$ for any QLS solver in the sparse access model.
\end{proof}

\section{Lower bound in condition number and sparsity}
\label{sec:sparsity}

Here we prove the query-complexity lower bound of QLS incorporating the sparsity $s$ via a reduction from the bounded-error \PARITYOR problem. Our proof uses the technique applied in Ref.~\cite{Low_2017} for proving the query-complexity lower bound of Hamiltonian simulation.

\begin{theorem}[Lower bound of QLS with sparsity]
\label{lem:qls_bounded}
    Let the {\rm QLS} be defined as Problem~\ref{prob:QLS} with the target error $\epsilon$ fixed to a constant. The query complexity of {\rm QLS} is $\Omega(\kappa \sqrt{s})$.
\end{theorem}

\begin{proof}
Our goal is to set up a QLS that solves the bounded \PARITYOR[n][m] problem, again given the sparse access model in Def.~\ref{def:A_oracle}. We begin by defining matrices
\begin{align}
    &A := I-\frac{1}{3\e^{1/n}}B, 
    \\
    &B := \sum_{t=1}^{n} \Big(
        H_t\otimes \ketbra{t+1}{t}_{\rm c}
        + I\otimes \ketbra{t+n+1}{t+n}_{\rm c}+ H_{n-t+1}\otimes \ketbra{t+2n+1}{t+2n}_{\rm c}
    \Big),\label{eq:B_def}
    \\
    &H_t:=\begin{pmatrix}
        I - C_t&C_t\\C_t&I - C_t
    \end{pmatrix} ,
    \quad C_t:=\begin{pmatrix}
        x_{t,1}&x_{t,2}&\dots&x_{t,m}\\
        x_{t,m}&x_{t,1}&\dots&x_{t,{m-1}}\\
        \vdots&\vdots&\ddots&\vdots\\
        x_{t,2}&x_{t,3}&\dots&x_{t,1}
    \end{pmatrix} .
\end{align}
The construction of $C_t$ requires queries to the oracle $\mathcal{O}^{(2)}_X$ in Eq.~\eqref{eq:access_party_or}.
As discussed in Section~\ref{sec:preliminaries}, we assume that at most one bit of each bit string $X_t=(x_{t,1},\dots,x_{t,m})$ is non-zero, so $\sum_{j=1}^mx_{t,j}$ is either 0 or 1 and indeed
\begin{align}
    f_{\OR}(X_t)=\sum_{j=1}^mx_{t,j}.
\end{align}
The promise may appear to imply the $m \times m$ matrix $C_t$ is $1$-sparse. However, we treat $C_t$ as an $m$-sparse matrix, because the binary variables $x_{t,j}$ are unknown when accessing them and thus the position oracle~\eqref{eq:position_oracle} needs to support all the entries. Similarly, matrix $I - C_t$ is treated as an $m$-sparse matrix.
In other words, we let $\mathcal{P}_A^{\rm pos}$ return column indices associated with these matrices even if the corresponding entries are zero.
It follows from these definitions that the matrix $A$ has sparsity $2m +1$.

The uniform superposition $\ket{u_m} = \frac{1}{\sqrt{m}}\sum_{j=1}^m \ket{j}$ is an eigenstate of $C_t$ that obeys
\begin{align}
    C_t\ket{u_m}=f_{\OR}(X_t)\ket{u_m},
\end{align}
and the application of $H_t$ to the state $\ket{k}\ket{u_m}$ with $k\in\{0,1\}$ gives
\begin{align}
    H_t\ket{k}\ket{u_m}=\ket{k\oplus f_{\OR}(X_t)}\ket{u_m}.
\end{align}

Using this property, we see that the application of $B$ to the state $\ket{k}\ket{u_m}\ket{t}_{\rm c}$ results in
\begin{align}
    B\ket{k}\ket{u_m}\ket{t}_{\rm c}
    =\begin{cases}
        \ket{k\oplus f_{\OR}(X_t)}\ket{u_m}\ket{t+1}_{\rm c}&\text{if }t\in[1,n],
        \\
        \ket{k}\ket{u_m}\ket{t+1}_{\rm c}&\text{if } t\in[n+1,2n],
        \\
        \ket{k\oplus f_{\OR}(X_{3n-t+1})}\ket{u_m}\ket{t+1}_{\rm c}&\text{if }t\in[2n+1,3n].
    \end{cases}
\end{align}

Next, a QLS solver with the initial state $\ket{0}\ket{u_m}\ket{1}_{\rm c}$ approximates the state
\begin{align}
\label{eq:inverse_parity_or}
\begin{split}
    \ket{\Psi}
    &:=\frac{A^{-1}\ket{0}\ket{u_m}\ket{1}_{\rm c}}{\lnorm*{A^{-1}\ket{0}\ket{u_m}\ket{1}_{\rm c}}}
    =
    \frac{\sum_{t=0}^{\infty}\frac{1}{3^t}\e^{-t/n}B^t\ket{0}\ket{u_m}\ket{1}_{\rm c}}{\lnorm*{A^{-1}\ket{0}\ket{u_m}\ket{1}_{\rm c}}}
    \\
    &=
    c\sum_{t=0}^{n-1}\frac{1}{3^t}\e^{-t/n}\Big(
        \ket{\oplus_{j=1}^t f_{\OR}(X_j)}\ket{u_m}\ket{t+1}_{\rm c}
        \\
        &\qquad
        + \e^{-1} \ket{\oplus_{j=1}^n f_{\OR}(X_j)}\ket{u_m}\ket{t+n+1}_{\rm c}
        \\
        &\qquad
        + \e^{-2} \ket{\oplus_{j=1}^{n-t} f_{\OR}(X_j)}\ket{u_m}\ket{t+2n+1}_{\rm c}
    \Big),
\end{split}
\end{align}
with $c=\left[\left(1-27^{-n}\e^{-3}\right)\lnorm*{A^{-1}\ket{0}\ket{u_m}\ket{1}_{\rm c}}\right]^{-1}$.
Here, an application of $H_t$ in $B$ writes $f_{\OR}(X_t)$ into the first register. Applying $H_t$ for each $t$ in the clock register calculates $\oplus_t f_{\OR}(X_t)$. To obtain the third line of the last expression, we used $H^2_t\ket{k}\ket{u_m}=\ket{k}\ket{u_m}$ in the subspace where the second register is fixed to the $\ket{u_m}$ state.

If we measure the clock register and obtain $t\in[n+1,2n]$, then the first register is in the state with the correct parity $\ket{\oplus_{t=1}^nf_{\OR}(X_t)}$. The probability of obtaining this outcome is
\begin{align}
    p
    =
    \frac{\e^{-2}}{1+\e^{-2}+\e^{-4}} \approx 0.117.
\end{align}
Similar to the proof of Theorem \ref{lem:qls_unbounded}, if we consider the protocol where we return the first register of $\ket{\Psi}$ when we measure $t\in[n+1,2n]$ and return 0 or 1 randomly otherwise, this protocol corresponds to Problem \ref{def:boolean_problem} with $\gamma=\frac{p}{2}=\mathrm{constant}$. 
This conclusion holds with a sufficiently small constant target error $\epsilon$ in preparing $\ket{\Psi}$, following a similar discussion as in the proof of Theorem \ref{lem:qls_unbounded}.
This means that a QLS solver with constant target error can solve the bounded-error \PARITYOR problem.

We have reduced the bounded-error \PARITYOR[n][m] problem to the QLS problem. From Lemma~\ref{lem:parity_or}, the former has query complexity of $\Theta(n\sqrt{m})$ in terms of access to $\mathcal{O}^{(2)}_X$.
Because the non-zero entries of the matrix $A$ that depend on the bit string $X=(X_1,\dots,X_n)$ can be expressed as an affine function of $x_{t,k}$, its sparse matrix access $\mathcal{P}_A$ can be implemented with at most two queries to controlled $\mathcal{O}^{(2)}_X$.
We give an explicit construction of $\mathcal{P}_A$ in Appendix~\ref{app:sparse_oracle}.
As a result, the query complexity of QLS is also lower-bounded by $\Omega(n\sqrt{m})$ in terms of access to $\mathcal{P}_A$.

Now we relate the parameters $n,m$ from the \PARITYOR problem to the condition number $\kappa$ and the sparsity $s$ of our matrix $A$.
The condition number of $A$ satisfies $\kappa \le \frac{1+\frac{1}{3}\e^{-1/n}\opnorm{B}}{1-\frac{1}{3}\e^{-1/n}\opnorm{B}}$, so we want to bound $\opnorm{B}$.
From Eq. \eqref{eq:B_def}, we have $\opnorm*{B^\dag B}=\opnorm[\big]{\sum_{t=1}^{n} (H_t^\dag H_t\otimes \ketbra{t}{t}_{\rm c} + I\otimes \ketbra{t+n}{t+n}_{\rm c}+ H_{n-t+1}^\dag H_{n-t+1}\otimes \ketbra{t+2n}{t+2n}_{\rm c})}=\max\{(\max_t \opnorm{H_t})^2,1\}$, so $\opnorm{B}=\sqrt{\opnorm{B^\dag B}}=\max\{\max_t \opnorm{H_t},1\}$.
Because $\opnorm{H_t}$ can be bounded by a constant as 
\begin{align}
    \opnorm{H_t}=\opnorm*{\begin{pmatrix}
        -1&1\\1&-1
    \end{pmatrix}\otimes C_t+I\otimes I} \le \opnorm*{\begin{pmatrix}
        -1&1\\1&-1
    \end{pmatrix}\otimes C_t} + \opnorm{I\otimes I} \le 3,
\end{align}
where we used $\|C_t\|\le1$ under the promise, the condition number is given as $\kappa \le 2/(1-\e^{-1/n}) = \Theta(n)$.
The sparsity of the matrix $A$ is $2m+1$, i.e., $s = \Theta(m)$.
Substituting both in $\Omega(n\sqrt{m})$ we obtain the lower bound $\Omega(\kappa\sqrt{s})$ for any QLS solver in the sparse-matrix oracle access model.
\end{proof}

\section{Discussion}
\label{sec:discussion}
We employed the reduction of the unbounded-error \PARITY problem to rederive the $\kappa$ and $\epsilon$ dependence of the query-complexity lower bound for QLS solvers.
We then turned to the sparsity $s$ dependence by invoking the known complexity bound of the bounded-error \PARITYOR problem.
One may wonder why we did not resort to the unbounded-error \PARITYOR problem to accommodate all three parameters $\epsilon$, $\kappa$, and $s$ in the lower bound, which would appear to achieve our ultimate goal.
The unbounded \OR[m] problem can be solved with $\Theta(1)$ queries~\cite{Montanaro_2011}. 
Unfortunately, this leads to $\Theta(n)$ query complexity for the unbounded \PARITYOR[n][m] problem as shown in Lemma~\ref{lem:parity_or_unbounded}. 
Recalling that the parameter $m$ is related to sparsity as $s = \Theta(m)$, a reduction of the unbounded problem to QLS would incorporate $\epsilon$ at the price of losing the sparsity dependence with the presented proof strategy.

To accommodate all three parameters in the lower bound, one potential avenue is to reduce QLS from a different problem, namely a composition of Boolean functions such that the unbounded-error query complexity is $\Omega(n\sqrt{m})$. To this end, tools developed in Refs.~\cite{Sherstov_2009} and~\cite{Montanaro_2011} could be useful.
If reductions from Boolean problems fail to yield the desired result, we must attack the problem from a completely different angle.
For example, in Ref.~\cite{Montanaro_2024} it was proved that the complexity lower bound for approximating entries of $g(A)$ to target error $\epsilon$ is given by the minimum degree of a polynomial that $\epsilon$-approximates function $g$ on the spectrum of $A$. Therefore, finding the minimum degree for $g(x)=1/x$ would give a lower bound for a problem closely related to QLS.
However, the $\epsilon$-dependence of the minimum degree is not known, leaving the $\epsilon$ dependence of the bound as an open question.
Moreover, we could not incorporate the $s$ dependence with this approach, as it does not appear in the degree.
We leave the joint lower bound in all three parameters for future work.

\section*{Acknowledgments}
We are grateful to Ansis Rosmanis and Jordi Weggemans for their feedback on the manuscript, and we thank Gabriel Marin-Sanchez and Enrico Rinaldi for helpful discussions.

\bibliography{bibliography}
\bibliographystyle{utphys}

\appendix

\section{Alternative proof of Lemma \ref{lem:parity_or_unbounded}}
\label{app:composition}
In this appendix, we provide an alternative proof of Lemma \ref{lem:parity_or_unbounded}.
\begin{proof}
First, we give a lower bound by considering the reduction from \PARITY[n]. Provided $n$ bits $(x_1,\dots,x_n)$, we pad each bit $x_t$ with $m-1$ zeros as $X_t:=(x_t,0,\dots,0)\in\{0,1\}^m$. Thus, it is reduced to a \PARITYOR[n][m] problem. Since solving \PARITYOR[n][m] also solves \PARITY[n], which has the tight bound of $\Theta(n)$, we obtain the lower bound $\Omega(n)$ for \PARITYOR[n][m]. 

To find an upper bound, we construct a sign-representation for a partial Boolean function $f$ and use $\mathsf{UQ}(f) = \ceil{\text{sdeg}(f)/2}$ to obtain its unbounded-error query complexity. The partial \OR function $f_{\OR}$ admits the following sign representation,
\begin{align}
    g_t(X_t):=\sum_{i=1}^m x_{t,i}-\frac{1}{2},
\end{align}
i.e., it satisfies $g_t(X_t)=\frac{1}{2}>0$ if $f_{\OR}=1$ and otherwise $g_t=-\frac{1}{2}<0$.
Then, we sign represent the function $f_{\PAR\text{-}\OR}$ as
\begin{align}
    h(X):=(-1)^{n+1}\prod_{t=1}^{n}g_t(X_t),
\end{align}
which satisfies $h(X)>0$ when $f_{\PAR\text{-}\OR}=1$ and otherwise $h(X)<0$. 
Therefore, there exists a quantum algorithm that solves unbounded-error \PARITYOR[n][m] with $\mathsf{UQ}(f_{\PAR\text{-}\OR}) = \ceil{\text{sdeg}(f_{\PAR\text{-}\OR})/2} \le \ceil{n/2}$ queries, which gives the upper bound $O(n)$.
Combining the obtained lower bound and upper bound, we arrive at the tight query complexity bound $\Theta(n)$.
\end{proof}

\section{Construction of the oracle}
\label{app:sparse_oracle}

In this appendix, we construct the sparse oracle $\mathcal{P}_A$ defined in Def.~\ref{def:A_oracle} to determine the number of $\mathcal{O}^{(2)}_X$ calls [Eq.~\eqref{eq:access_party_or}] required for a single $\mathcal{P}_A$ call.
As the position of non-zero entries of $A$, constructed from the dense matrices $H_t$, is independent of the bit string $X$, $\mathcal{P}_{A}^{\rm pos}$ does not use $\mathcal{O}^{(2)}_X$.
Therefore, we only consider the construction of $\mathcal{P}_A^{\rm val}$ in this appendix.
Moreover, because we are only interested in the number of calls to $\mathcal{O}^{(2)}_X$, we focus on $\mathcal{P}_B^{\rm val}$. 
To obtain $\mathcal{P}_A^{\rm val}$, we only need to add diagonal entries of 1 and multiply $\mathcal{P}_B^{\rm val}$ by $-\frac{1}{3}\e^{-1/n}$ using arithmetic operations.

Recalling that the matrix acts on three registers, our problem reduces to that of constructing the  oracles
\begin{align}
\label{eq:P_B^val}
\mathcal{P}_B^{\rm val} \ket{k,i,t}\ket{k',j,t'} \ket{z} = \ket{k,i,t}\ket{k',j,t'} \ket{z \oplus B_{(k,i,t),(k',j,t')}},
\end{align}
from constant calls to $\mathcal{O}^{(2)}_X$.

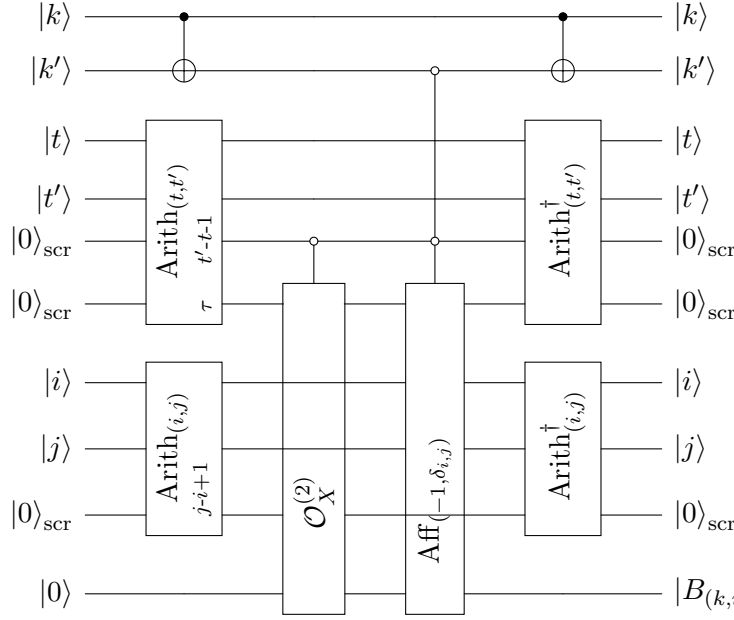
\begin{figure}
\centering
\begin{quantikz}[row sep=0.5cm, column sep=0.8cm, thin lines,transparent]
    \lstick{$\ket{k}$} 
    & \ctrl{1} 
    &
    &
    & \ctrl{1}
    &\rstick{$\ket{k}$}
    \\
    \lstick{$\ket{k'}$} 
    & \targ{}
    &  
    & \octrl{3} 
    & \targ{}
    &\rstick{$\ket{k'}$}
    \\
    \lstick{$\ket{t}$} 
    & \gate[4, label style={xshift=-0.15cm}, disable auto height][1cm]{\rotatebox[origin=l]{90}{Arith$_{(t,t')}$}}
    &
    &  
    & \gate[4, disable auto height][1cm]{\rotatebox[origin=l]{90}{Arith$^\dag_{(t,t')}$}}
    &\rstick{$\ket{t}$}
    \\
    \lstick{$\ket{t'}$} 
    &
    &  
    &  
    & 
    &\rstick{$\ket{t'}$}
    \\
    \lstick{$\ket{0}_{\rm scr}$} 
    & \gateoutput{\rotatebox[origin=c]{90}{$t'$-$t$-$1$}}
    & \octrl{1}
    & \octrl{1} 
    &  
    & \rstick{$\ket{0}_{\rm scr}$} 
    \\
    \lstick{$\ket{0}_{\rm scr}$} 
    & \gateoutput{\rotatebox[origin=c]{90}{$\tau$}}
    & \gate[5, label style={yshift=-0.7cm}, disable auto height][0.8cm]{\rotatebox[origin=c]{90}{$\mathcal{O}^{(2)}_X$}} 
    & \gate[5,label style={yshift=-0.7cm}, disable auto height][0.77cm]{\rotatebox[origin=l]{90}{${\rm Aff}_{(-1,\delta_{i,j})}$}}
    &
    &\rstick{$\ket{0}_{\rm scr}$}
    \\
    \lstick{$\ket{i}$} 
    & \gate[3, label style={xshift=-0.15cm}, disable auto height][1cm]{\rotatebox[origin=l]{90}{Arith$_{(i,j)}$}}
    & \linethrough
    & 
    & \gate[3, disable auto height][1cm]{\rotatebox[origin=l]{90}{Arith$^\dag_{(i,j)}$}}
    &\rstick{$\ket{i}$}
    \\
    \lstick{$\ket{j}$} 
    &
    & \linethrough    
    & 
    & 
    &\rstick{$\ket{j}$}
    \\
    \lstick{$\ket{0}_{\rm scr}$} 
    & \gateoutput{\rotatebox[origin=l]{90}{$j$-$i$+$1$}}
    &  
    & \linethrough 
    &
    & \rstick{$\ket{0}_{\rm scr}$}
    \\
    \lstick{$\ket{0}$} 
    &
    &  
    &  
    & 
    &\rstick{$\ket{B_{(k,i,t),(k',j,t')}}$}
\end{quantikz}
\caption{Quantum circuit for implementing $\mathcal{P}_B^{\rm val}$~\eqref{eq:P_B^val}. The registers with subscript ``scr'' serve as scratch registers that store intermediate results and are uncomputed at the end. See the main text for the definition of each box representing a quantum operation. In the circuit diagram, we omit the final write-out to $\ket{z}$ followed by uncomputation.}
\label{fig:circ_P_A}
\end{figure}

We initialize register $\ket{z} = \ket{0\dots0}$ where the number of qubits depends on the desired precision. Note that only for $t' = t+1 \pmod{3n}$ the entries $B_{(k,i,t),(k',j,t')}$ are non-zero. In other words, for all $t' \neq t+1$, the register $\ket{z}$ already contains the correct value. Thus we have 
\begin{align}   
    \mathcal{P}_B^{\rm val}\ket{k,i,t}\ket{k',j,t+1}\ket{0}=\begin{cases}\ket{k,i,t}\ket{k',j,t+1}\ket{(H_t)_{(k,i),(k',j)}}&\text{if } t\in [1, n],\\
    \ket{k,i,t}\ket{k',j,t+1}\ket{I_{(k,i),(k',j)}}&\text{if }t\in [n+1, 2n],\\
    \ket{k,i,t}\ket{k',j,t+1}\ket{(H_{3n-t+1})_{(k,i),(k',j)}}&\text{if }t\in[2n+1, 3n].
    \end{cases}
\end{align}
Next, the elements of $H_t$ are expressed as

\begin{align}
    \ket{k,i}\ket{k',j}\ket{(H_t)_{(k,i),(k',j)}}
    =\begin{cases}
    \ket{k,i}\ket{k\oplus1,j}\ket{x_{t,j-i+1}} & \text{if } k'=k \oplus 1, \\
    \ket{k,i}\ket{k,j}\ket{\delta_{i,j}-x_{t,j-i+1}} & \text{if } k' = k.
    \end{cases}
\end{align}

Now we construct a circuit to implement this.
We first use arithmetic operations, Arith$_{(t,t')}$ in Fig.~\ref{fig:circ_P_A}, to prepare $\ket{t'-t}$ and $\ket{\tau}$ from $\ket{t}\ket{t'}$, where
\begin{align}
    \tau=\begin{cases}
        t&\text{if } t\in[1,n],\\
        0&\text{if } t\in[n+1,2n],\\
        3n-t+1&\text{if } t\in[2n+1,3n],
    \end{cases}
\end{align}
and $\ket{j-i+1}$ from $\ket{i}\ket{j}$ (Arith$_{(i,j)}$ in Fig.~\ref{fig:circ_P_A}).
We define $x_{0,j}=0$ for all $j$ and extend the definition of $\mathcal{O}_X^{(2)}$ to include $t=0$.
Then we can construct a circuit as in Fig.~\ref{fig:circ_P_A}, where $\mathrm{Aff}_{(-1,\delta_{i,j})}$ represents an arithmetic operation that performs the affine mapping $x_{\tau,j-i+1}\mapsto \delta_{i,j}-x_{\tau,j-i+1}$ conditioned on the control register $t'-t-1=0$, $\tau\ne0$ and $k'=k$. The Arith$^\dag$ boxes on the last layer in Fig.~\ref{fig:circ_P_A} uncompute the scratch registers.
Therefore, we can construct $\mathcal{P}_B^{\rm val}$ with two queries to controlled $\mathcal{O}^{(2)}_X$.

\end{document}